\documentclass[authoryear,11pt]{article}
\usepackage[T1]{fontenc}
\usepackage[ansinew]{inputenc}
\usepackage{url}
\usepackage{natbib}
\usepackage[english]{babel}
\usepackage{amsmath}
\usepackage{latexsym}
\usepackage{amssymb}
\usepackage{subfigure}
\usepackage{graphicx}

\usepackage{comment}
\includecomment{figurepdf}

\newcommand{\mZ}{Z}
\newcommand{\mM}{M}
\newcommand{\T}{\top}
\newcommand{\pr}{\mathrm{pr}}

\newcommand{\G}{\mathcal{G}}

\newcommand{\dcup}{\,\dot{\cup}\,}
\newcommand{\I}{\mathcal{I}}
\newcommand{\J}{\mathcal{J}}

\newcommand{\Mob}{M\"{o}bius}
\newcommand{\LML}{LML}
\newcommand{\ci}{\mbox{$\perp \! \! \! \perp$}}

\usepackage{amsthm}
\theoremstyle{plain}
\newtheorem{thm}{Theorem}[section]
\newtheorem{lem}[thm]{Lemma}

\newtheorem{cor}[thm]{Corollary}

\theoremstyle{definition}
\newtheorem{defn}{Definition}[section]

\newtheorem{exmp}{Example}[section]

\theoremstyle{remark}

\usepackage{color}
\definecolor{colorS}{rgb}{0.7695312,0.0468750,0.1601562}
\begin{document}
\title{Dichotomization invariant log-mean linear parameterization for discrete graphical models of marginal independence}
\author{Alberto Roverato\\ University of Bologna\\
\texttt{alberto.roverato@unibo.it}}
\date{September 12, 2013}
\maketitle
\begin{abstract}\small
We extend the log-mean linear parameterization introduced by \citet{roverato2013log} for binary data to discrete variables with arbitrary number of levels, and show that also in this case it can be used to parameterize bi-directed graph models. Furthermore, we show that the log-mean linear parameterization allows one to simultaneously represent marginal independencies among variables and marginal independencies that only appear when certain levels are collapsed into a single one. We illustrate the application of this property by means of an example based on genetic association studies involving single-nucleotide polymorphisms. More generally, this feature provides a natural way to reduce the parameter count, while preserving the independence structure, by means of substantive constraints that give additional insight into the association structure of the variables.
\end{abstract}


\section{Introduction}\label{SEC:introduction}
Graphical models of marginal independence use a graph where every vertex is associated with a variable and missing edges encode  marginal independence relationships according to a given Markov property; see \citet{pearl1994,kauermann1996dualization,banerjee2003dualization,ric:2003}.
These models were introduced by \citet{cw:1993,cw:1996} as \emph{covariance graphs}, with dashed lines to represent edges in the graph. More recently, lines with two arrowheads are often used in place of dashed
edges and, accordingly, these models are also referred to as  \emph{bi-directed graph models}. Graphical models of marginal independence have appeared in several applied contexts as described in \citet{drton2008binary} and references therein. Their application is typically suggested when the observed variables are jointly affected by unobserved variables \citep[see, among others,][]{ric:2003,maathuis2009estimating,colombo2012learning} and, furthermore, they are a special case of \emph{acyclic directed mixed graphs} \citep{ric:2003} and can be regarded as the building blocks  of \emph{regression graph models} \citep{drton2009discrete,wermuth2012sequencies}.

The probability distribution of a set of discrete variables is characterized by the associated probability table, but defining a suitable parameterization for bi-directed graph models is not straightforward; see \citet{drton2008binary,drton2009discrete,lup-mar-ber:2009,roverato2013log}. A basic requirement for the flexible implementation of marginal constraints is that interaction terms involving a subset of variables satisfy \emph{upward compatibility}, that is they should reflect a property of the corresponding marginal distribution; see \citet{drton2008binary} for details and extensive references. Upward compatibility means invariance with respect to marginalization but, for discrete variables with arbitrary number of levels, a stronger invariance property may be required. As shown in the example below, there are situations where the research question involves different collapsed versions of a same variable. Collapsing two or more levels of a discrete variables into a single level can be regarded as as a special kind of marginalization and invariance with respect to this operation is an useful feature for a parameterization.
\begin{exmp}[\emph{Genetic association analysis}]\label{EXA:genetic001}
Genetic association studies aim at identifying genetic factors associated with a certain phenotype, such as a disease; see \citet{balding2006tutorial} for a review of statistical approaches to population association studies. Single-Nucleotide Polymorphisms (SNPs) are the most common form of variation in the human genome \citep[see][]{hirschhorn2005genome}. A SNP is a change in one nucleotide at a given genomic position. Commonly, SNPs are diallelic with two of the four bases $A$ (adenine), $C$ (cytosine), $T$ (thymine) and $G$ (guanine) occurring at the considered locus where it is possible that a \emph{wild} allele, $W$, is substituted by a \emph{mutant} allele, $M$. Hence, a SNP has three possible genotypes, $WW$, $WM$ and $MM$, and can be represented as a three-level discrete variable. The latter representation makes it possible to identify the \emph{codominant genotype effect} of the SNP on the phenotype, but the relevant phenotype may also be associated with alternative representations of the SNP. In the \emph{dominant genotype model}, heterozygote individuals are expected to have the same phenotype as $MM$ homozygote individual so that the levels $WM$ and $MM$ are collapsed into a single one to give $WM+MM$ vs. $WW$. Conversely, in the \emph{recessive genotype model} the three levels are dichotomized as $WM+WW$ vs. $MM$. In general, none of these three genetic models for a specific SNP is favored a priori, and the research question also concerns the identification of the most appropriate representation of a SNP. Clearly, there is a loss of efficiency in fitting a different statistical model for every possible genetic model, and this is especially true when more SNPs are simultaneously considered.
\end{exmp}

In this paper we extend the \emph{Log-Mean Linear (\LML)}  parameterization introduced by \citet{roverato2013log} for binary data to discrete variables with arbitrary number of levels and show that also in this case it can be used to parameterize bi-directed graph models. Furthermore, we show that the \LML\ parameterization satisfies a stronger version of upward compatibility that we call \emph{dichotomization invariance}. Every \LML\ parameter can be uniquely associated with a cell either of the cross-classified table or of a marginal table and it is invariant with respect to both marginalization and collapsing operations that does not involve such a cell. In this way, the \LML\ parameterization allows one to simultaneously represent marginal independencies among variables and marginal independencies that only appear when certain levels are collapsed into a single one. This feature is useful in several applied contexts, but it also provides a natural way to reduce the parameter count, while preserving the independence structure, by means of substantive constraints that give additional insight into the dependence structure of variables. Being able to implement additional substantive constraints so as to specify parsimonious submodels is a key issue in graphical modelling and, more generally, in multivariate analysis  \citep[see, for instance,][]{hojsgaard2008graphical} but it is especially relevant in marginal modelling because the number of parameters in a bi-directed graph model can be relatively large even for sparse graphs \citep{richardson2009factorization,eva-rich:2013}.

This paper is organized as follows. Section~\ref{SEC:background.paramet} provide a review of the theory of discrete bi-directed graph models and of the associated parameterizations as required for this paper. Section~\ref{SEC:poly.Moebius} contains the extension of the \LML\ parameterization to the general discrete case, whereas in Section~\ref{SEC:B-expanded} we introduce the binary expansion operation and state the connected set Markov property for $B$-expanded graphs. The genetic association analysis of Example~\ref{EXA:genetic001} is used throughout the paper to illustrate the main ideas whereas two applications are given in Section~\ref{SEC:applications}. Finally, Section~\ref{SEC:discussion} contains a brief discussion. Proofs are deferred to the Appendix.

\section{Parametrizations of discrete bi-directed graph models}\label{SEC:background.paramet}

\subsection{Bi-directed graph models}\label{SUB.luc.bigraphs}
Let $Y_{V}=(Y_{v})_{v\in V}$ be a  random vector with entries indexed by a finite set $V$ with $|V|=p$. In graphical models of marginal independence every variable $Y_{v}$, with $v\in V$, is associated with a vertex of a \emph{bi-directed graph} $\G=(V, E)$. The edge set $E$ is a collection of unordered, distinct, pairs of vertices and every edge $\{i, j\}\in E$ is represented as a line with two arrowheads, $i\leftrightarrow j$. A graph is \emph{complete} if every pair of vertices is joined by an edge. A subset $\emptyset\neq U\subseteq V$ induces a \emph{subgraph} $\G_{U}=(U, E_{U})$ where $E_{U}=E\cap (U\times U)$. If $\G_{U}$ is disconnected we say that $U$ is a \emph{disconnected set} of $\G$ and denote by $C_{1},\ldots, C_{r}$ its inclusion maximal connected sets that we call the \emph{connected components} of $U$; see \citet{ric:2003} for details. Recall that $U=C_{1}\dcup\cdots\dcup C_{r}$ where the symbol $\dcup$ denotes a union of disjoint sets.

A bi-directed graph model is the family of probability distributions for $Y_{V}$ that satisfy a given Markov property with respect to a bi-directed graph $\G=(V, E)$. The distribution of $Y_{V}$ is said to satisfy the \emph{pairwise Markov property} with respect to $\G$ if for every $\{i,j\}\not\in E$, with $i\neq j$, it holds that $Y_{i}$ is independent of $Y_{j}$; in symbols $Y_{j}\ci Y_{k}$. The distribution of $Y_V$ satisfies the \emph{global Markov property}, also called the \emph{connected set Markov property} by \citet{ric:2003}, if for every disconnected set $U$ of $\G$, the subvectors corresponding to its connected components $Y_{C_1},\dots,Y_{C_r}$ are mutually independent; $Y_{C_1}\ci\cdots\ci Y_{C_r}$ or, equivalently, $C_{1}\ci\cdots\ci C_{r}$. We remark that the connected set Markov property implies the pairwise
Markov property, whereas the converse is not true in general, even for strictly positive distributions.
\begin{exmp}[\emph{bi-directed 4-chain}]
In the bi-directed graph of Figure~\ref{FIG.bid.graph} the disconnected sets are $\{1,3\}$, $\{1,4\}$, $\{2,4\}$, $\{1,2,4\}$ and $\{1,3,4\}$. Under the connected set Markov property, the sets $\{1,2,4\}$ and $\{1,3,4\}$ encode the independencies $X_{\{1,2\}} \ci X_{4}$ and $X_{1} \ci X_{\{3,4\}}$, respectively, and these imply the independencies encoded by each of the remaining disconnected sets.
\end{exmp}
\begin{figure}
\begin{center}
\begin{figurepdf}
\includegraphics[scale=1]{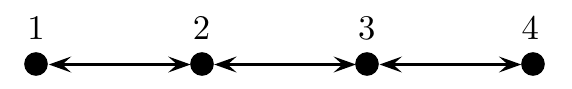}
\end{figurepdf}
\end{center}
\caption{\small Bi-directed 4-chain.}
\label{FIG.bid.graph}
\end{figure}

\subsection{The \Mob\ and \LML\ parameters for the binary case}
In order to highlight that a variable is binary we denote it by $X$ and, without loss of generality, we assume that $X\in \{0, 1\}$ so that $X_{V}=(X_{v})_{v\in V}$ is a multivariate Bernoulli random vector taking values in $\{0,1\}^{p}$. From the fact that $\{0, 1\}^{p}=\{(1_{U}, 0_{V \backslash U})\mid U\subseteq V\}$, it follows that one can write the probability distribution of $X_{V}$ as a vector $\pi=(\pi_{U})_{U \subseteq V}$ with entries $\pi_{U}=P(X_{U}=1_{U},X_{V\backslash U}=0_{V\backslash U})$.

The multivariate Bernoulli distribution belongs to the natural exponential family with mean parameter $\mu=(\mu_{U})_{U \subseteq V}$ where
\begin{eqnarray}\label{EQN:mobius.binary}
\mu_{\emptyset}=1\quad\mbox{and}\quad \mu_{U}=\pr(X_{U}=1_{U})\quad \mbox{for every } U\subseteq V\mbox{ with }U\neq\emptyset.
\end{eqnarray}
The mean parameter $\mu$ was called the \emph{\Mob\ parameter} by \citet{drton2008binary} because $\mu_U=\sum_{E \subseteq V\backslash U} \pi_{U\cup E}$ for every $U\subseteq V$ so that the inverse map $\mu\mapsto\pi$ can be computed by \emph{\Mob\ inversion} as $\pi_U=\sum_{E \subseteq V\backslash U} (-1)^{|E|} \mu_{U\cup E}$ for every $U \subseteq V$; see \citet[][Appendix~A]{lau:1996}. Let $\mZ$ and $\mM$ be two $(2^p \times 2^p)$ matrices with entries indexed by the subsets of $V \times V$ and given by $\mZ_{U,H}= 1 (U \subseteq H)$ and $\mM_{U,H}= (-1)^{|H \backslash U|} 1(U \subseteq H)$, respectively, where $1(\cdot)$ denotes the indicator function. Then, one can write the linear relationship between $\pi$ and $\mu$ in matrix form as $\mu=\mZ\pi$ and $\pi=\mM\mu$ and \Mob\ inversion follows by noticing that $\mM=\mZ^{-1}$.

\citet{roverato2013log} introduced the Log-Mean Linear (\LML) parameter $\gamma=(\gamma_{U})_{U \subseteq V}$ whose entries are computed as a log-linear expansion of the \Mob\ parameters
\begin{eqnarray}\label{EQN:map.gamma}
\gamma_U=\sum_{E \subseteq U} (-1)^{|U\backslash E|} \log \mu_E\quad\mbox{for every } U\subseteq V
\end{eqnarray}
so that, in matrix form, $\gamma=\mM^{\T}\log \mu$ where $\mM^{\T}$ denotes the transpose of $\mM$. Note that $\gamma_{\emptyset}=0$ and, furthermore, that $\pi$ can be analytically computed by applying \Mob\ inversion twice to obtain $\pi= \mM \exp \mZ^{\T} \gamma$. The \LML\ parameter can be regarded as a linearization  of the \Mob\ parameter because certain multiplicative constraints in $\mu$ correspond to \LML\ interactions equal to zero as follows.
\begin{thm}[Theorem~1 in \citet{roverato2013log}]\label{THM:gamma.eq.zero.and.fact.mu}
Let $X_{V}$ be a vector of binary variables with \Mob\ and \LML\ parameters $\mu$ and $\gamma$, respectively. Then, for a pair of disjoint, nonempty, proper subsets $A$ and $B$ of $V$, the following conditions are equivalent:
\begin{itemize}
\item[(i)] $X_{A}\ci X_{B}$;
\item[(ii)] $\mu_{A^{\prime}\cup
B^{\prime}}=\mu_{A^{\prime}}\times \mu_{B^{\prime}}$ for every
$A^{\prime}\subseteq A$ and $B^{\prime}\subseteq B$;
\item[(iii)]
$\gamma_{A^{\prime}\cup B^{\prime}}=0$ for every $A^{\prime}\subseteq
A$ and $B^{\prime}\subseteq B$ such that $A^{\prime}\neq\emptyset$
and $B^{\prime}\neq \emptyset$.
\end{itemize}
\end{thm}
The equivalence (i)$\Leftrightarrow$(ii) of Theorem~\ref{THM:gamma.eq.zero.and.fact.mu} follows immediately
from Theorem~1 of \citet{drton2008binary} where it is shown that, in the binary case, bi-directed graph models can be parameterized by imposing multiplicative constraints on the \Mob\ parameters. The latter result can be expressed in terms of the \LML\ parameters as follows.
\begin{thm}[Theorem~2 in \citet{roverato2013log}]\label{THM:gamma.eq.zero.and.disc.MP}
Let $\gamma=(\gamma_{U})_{U\subseteq V}$ be the \LML\ parameter of the binary random vector $X_{V}$ and let $\G=(V, E)$ be a bi-directed graph. The distribution of $X_{V}$ satisfies the connected set Markov property with respect to $\G$ if and only if for every set $U\subseteq V$ that is disconnected in $\G$ it holds that $\gamma_{U}=0$.
\end{thm}
\begin{exmp}[\emph{bi-directed 4-chain  cont.}]\label{EXA:bi-directed-4-chain-2}
The disconnected sets of the
bi-directed graph of Figure~\ref{FIG.bid.graph} are $\{1,3\}$, $\{1,4\}$, $\{2,4\}$, $\{1,2,4\}$ and $\{1,3,4\}$ so that $X_{V}$ satisfies the connected set Markov property with respect to such graph if and only if $\gamma_{\{1,3\}}=\gamma_{\{1,4\}}=\gamma_{\{2,4\}}=\gamma_{\{1,2,4\}}=\gamma_{\{1,3,4\}}=0$.
\end{exmp}

Point (iii) of Theorem~\ref{THM:gamma.eq.zero.and.fact.mu} is proved in \citet{roverato2013log} under the assumption that $\pi$ is strictly positive, and this implies that $\gamma_{U}$ is well-defined for every $U\subseteq V$. Here, it is worth remarking that, for every $U\subseteq V$, $\gamma_{U}$ in (\ref{EQN:map.gamma}) is well-defined if and only if $\mu_{U}>0$ because $\mu_{E}\geq \mu_{U}$ for every $E\subseteq U$.

\subsection{The general discrete case}
In this section, we consider the general case where the variables in $Y_{V}$ take on an arbitrary number of levels that we label as $I_{v}=\{0_{v}, 1_{v},\ldots, d_{v}\}$, for every $v\in V$. Hence, the state space of $Y_{V}$ is the product set $\I_{V}=${\Large $\times$}$_{v\in V} I_{v}$ that, with a slight abuse of terminology, we call a \emph{cross-classified table}. Accordingly, the elements $i\equiv i_{V} \in \I_V$ are called the \emph{cells} of the table. The probability distribution of $Y_{V}$ is characterized by the probability table $\varpi=(\varpi_{i})_{i\in\I_{V}}$, which we assume to be strictly positive. We remark that the symbol $\varpi$ is used to denote the probability table of an arbitrary discrete random vector whereas $\pi$ is used only in the binary case.

In the general discrete case, two different parameterizations of bi-directed graph models are available. \citet{lup-mar-ber:2009} showed that there exists a connection between bi-directed graph models and the \emph{multivariate logistic parameterization} of \citet{glo-mcc:1995}. Specifically, every multivariate logistic interaction is a log-linear parameter computed in the relevant marginal distribution and a bi-directed graph model can be specified by setting to zero all the interactions corresponding to the disconnected sets of the graph; see also \citet{rud-al:2010} and \citet{eva-rich:2013}. Parsimony can be achieved by setting further interactions to zero, but such additional constraints are typically difficult to interpret.

\citet{drton2009discrete} generalized the \Mob\ parameters to include non-binary variables and regression graph models. From the state space of $Y_{v}$, \citet{drton2009discrete} introduced the \emph{restricted state space} defined as  $J_{v}=I_{v}\backslash \{0_{v}\}=\{1_{v},\ldots, d_{v}\}$ so that the restricted state space of $Y_{V}$ is given by $\J_{V}=${\Large $\times$}$_{v\in V} J_{v}$. Hereafter, we refer to \lq\lq$0_{v}$\rq\rq{} as to the \emph{baseline} level of $Y_{v}$ and remark that the choice of the level to be set as baseline is arbitrary. For every $U\subseteq V$, with $U\neq\emptyset$, we denote by $\I_{U}$ and $\J_{U}$ the state space and the restricted  state space, respectively, of $Y_{U}$. Furthermore, for every $j\in \J_{V}$ we denote by $j_{U}$ the subset of levels of $j$ corresponding to the entries of $Y_{U}$, so that we can write $j_{U}\subseteq j$, and it holds that $\J_{U}=\{j_{U}\mid j\in\J\}$. Finally, when $U=\emptyset$ we use the convention that $j_{U}=\J_{U}=\emptyset$.

The \emph{saturated \Mob\ parameter} of $Y_{V}$ was defined by \citet{drton2009discrete} as the collection of marginal probabilities
\begin{eqnarray}\label{EQN:mobius.parameter.poly}
\pr(Y_{U}=j_{U})\quad\mbox{for every }j\in \J_{V}\mbox{ and }U\subseteq V\mbox{ with } U\neq\emptyset.
\end{eqnarray}
\citet{drton2009discrete} showed the saturated \Mob\ parameters characterize the distribution of $Y_{V}$ and that every bi-directed graph model is defined by an appropriate choice of multiplicative constraints on the saturated \Mob\ parameters. However, it is not clear how parsimony can be achieved with this parameterization. We also recall that the saturated \Mob\ parameters are closely related to the dependence ratios of \citet{Ekh-al:2000}.

\section{The \LML\ parameterization for the general discrete case}\label{SEC:poly.Moebius}
In this section we extend the \LML\ parameterization of \citet{roverato2013log} to discrete random variables with arbitrary number of levels and show that, also in this case, every bi-directed graph model can be defined  by setting certain \LML\ interactions to zero.

For every $v\in V$ and $i_{v}\in I_{v}$ we introduce the Bernoulli random variable $X_{v}^{i_{v}}$ defined as
\begin{eqnarray}\label{EQN:def.dichotomization}
X_{v}^{i_{v}}=\left\{
                  \begin{array}{ll}
                    1 & \mbox{if } Y_{v}=i_{v}  \\
                    0 & \mbox{otherwise.}
                  \end{array}
                \right.
\end{eqnarray}
In this way, every cell $i\in \I_{V}$ is associated with the random vector $X_{V}^{i}=(X_{v}^{i_{v}})_{v\in V}$ which follows a multivariate Bernoulli distribution, and we denote by $\pi^{i}=(\pi_{U}^{i})_{U\subseteq V}$ the  corresponding probability parameter. Accordingly, the mean and \LML\ parameter of $X_{V}^{i}$ can be computed as $\mu^{i}=\mZ\pi^{i}$ and $\gamma^{i}=\mM^{\T}\log \mu^{i}$, respectively.

It is straightforward to see,  from (\ref{EQN:mobius.binary}) and (\ref{EQN:def.dichotomization}), that
\begin{eqnarray}\label{EQN:mu.content}
\mu^{i}_{U}=\pr(X^{i}_{U}=1_{U})=\pr(Y_{U}=i_{U})\quad\mbox{for every }i\in \I_{V}\mbox{ and }U\subseteq V\mbox{ with } U\neq\emptyset,
\end{eqnarray}
and it follows that the \Mob\ parameter in (\ref{EQN:mobius.parameter.poly}) can be written as a collection of vectors of \Mob\ parameters for binary vectors; formally $\mu^{j}$ for $j\in \J_{V}$. In this way one can see that the collection of \LML\ parameters $\gamma^{j}$, for $j\in \J_{V}$, parameterizes the distribution of $Y_{V}$, because there exists a bijective map between $\gamma^{j}$ and $\mu^{j}$ for every $j\in \J_{V}$.
For every $i,i^{\prime}\in \I_{V}$ and $U\subseteq V$ such that $i_{U}=i^{\prime}_{U}$ it holds both that $\mu_{U}^{i}=\mu_{U}^{i^{\prime}}$ and that $\gamma_{U}^{i}=\gamma_{U}^{i^{\prime}}$ and, to remove redundancies, we write $\mu^{i_{U}}=\mu^{i}_{U}$  and $\gamma^{i_{U}}=\gamma^{i}_{U}$ so that the \Mob\ and the \LML\ parameters of $Y_{V}$ can be formally defined as
\begin{eqnarray*}
\mu=(\mu^{j_{U}})_{j\in \J_{V}, U\subseteq V}\quad\mbox{and}\quad \gamma=(\gamma^{j_{U}})_{j\in \J_{V}, U\subseteq V},
\end{eqnarray*}
respectively.

The fact that both the \Mob\ and the \LML\ parameter for the general case are defined as the collection of \Mob\ and \LML\ parameters, respectively, for the collection of  binary vectors $X^{j}_{V}$ with $j\in {\mathcal J}_{V}$ represents a key feature that confers useful properties to our approach. For instance, the generalization of relevant properties for these parameterizations follows immediately from the iterative application, for every $j\in {\mathcal J}_{V}$, of the corresponding properties of binary vectors. This is the case of the result of \citet{roverato2013log}, given in Theorem~\ref{THM:gamma.eq.zero.and.fact.mu}.
\begin{thm}\label{THM:poly.zero.gamma.and.fact.mu}
Let $\mu$ and $\gamma$ the \Mob\ and \LML\ parameters of $Y_{V}$, respectively. Then for a pair of disjoint, nonempty, proper subsets $A$ and $B$ of $V$, the following conditions are equivalent:
\begin{itemize}
\item[(i)] $Y_{A}\ci Y_{B}$;

\item[(ii)] $\mu^{j_{A^{\prime}\cup B^{\prime}}}=\mu^{j_{A^{\prime}}}\times \mu^{j_{B^{\prime}}}$  for every $j\in \J_{V}$,
$A^{\prime}\subseteq A$ and $B^{\prime}\subseteq B$;

\item[(iii)] $\gamma^{j_{A^{\prime}\cup B^{\prime}}}=0$ for every $j\in \J_{V}$, $A^{\prime}\subseteq A$ and $B^{\prime}\subseteq B$ such that $A^{\prime}\neq\emptyset$ and $B^{\prime}\neq \emptyset$.
\end{itemize}
\end{thm}
\begin{proof}
See the Appendix~\ref{APP:00B}.
\end{proof}
Theorem~\ref{THM:poly.zero.gamma.and.fact.mu} can be applied to show that bi-directed graph models can be parameterized by setting to zero the \LML\ interactions corresponding to the disconnected sets of $\G$, generalizing in this way the result of \citet{roverato2013log} given in Theorem~\ref{THM:gamma.eq.zero.and.disc.MP}.
\begin{cor}\label{THM:gamma.eq.zero.and.disc.MP.varpi}
Let $\gamma=(\gamma^{j_{U}})_{j\in \J_{V}, U\subseteq V}$ be the \LML\ parameter of $Y_{V}$ and let $\G=(V, E)$ be a bi-directed graph. The distribution of $Y_{V}$ satisfies the connected set Markov property with respect to $\G$ if and only if for every set $U\subseteq V$ that is disconnected in $\G$ it holds that $\gamma^{j_{U}}=0$ for every $j_{U}\in \J_{U}$.
\end{cor}
\begin{proof}
See the Appendix~\ref{APP:00B}.
\end{proof}
\begin{exmp}[\emph{bi-directed 4-chain cont.}] As in Example~\ref{EXA:bi-directed-4-chain-2}, we consider the graph $\G$ in Figure~\ref{FIG.bid.graph}. Then, by Corollary~\ref{THM:gamma.eq.zero.and.disc.MP.varpi}, the distribution of $Y_{V}$ satisfies the connected set Markov property with respect to $\G$ if and only if it holds that $\gamma^{j_{U}}=0$ for every $U\in \{\{1,3\}, \{1,4\}, \{2,4\}, \{1,2,4\}, \{1,3,4\}\}$ and $j_{U}\in \J_{U}$, i.e., if and only if $\gamma^{j_{\{1,3\}}}=\gamma^{j_{\{1,4\}}}=\gamma^{j_{\{2,4\}}}=\gamma^{j_{\{1,2,4\}}}=\gamma^{j_{\{1,3,4\}}}=0$ for every $j\in \J_{V}$. Clearly, if $Y_{V}$ is binary then $|\J_{V}|=1$ and the zero \LML\ interactions are the same as in Example~\ref{EXA:bi-directed-4-chain-2}.
\end{exmp}

The \emph{discrete bi-directed graph model} for $Y_{V}$ with graph $\G=(V,E)$ is defined as the set of positive probability distributions for $Y_{V}$ that obey the connected set Markov property with respect to $\G$. \citet[][Corollary~10]{drton2009discrete} showed that discrete bi-directed graph models are curved exponential families. The simple nature of the mapping $\varpi\mapsto\gamma$ allows one to see that $\gamma$ is a smooth parameterization of the saturated model and therefore that any model defined by imposing linear constraints on the \LML\ parameter of the saturated model is a curved exponential family. The family of submodels defined by linear constraints on $\gamma$ includes discrete bi-directed graph models by Corollary~\ref{THM:gamma.eq.zero.and.disc.MP.varpi}, and in the next section we will see how additional zero constraints on the \LML\ parameters can be specified so as to obtain interpretable bi-directed graph submodels.

\section{Dichotomization invariance and \boldmath $B$-expanded graphs}\label{SEC:B-expanded}
%
The \LML\ parameterization is based on the collection of binary vectors $X_{V}^{j}$ for $j\in \J$ and it is useful to take a closer look at how these parameters are computed. For every $i\in \I_{V}$ the computation of $\mu^{i}$ is based on the probability table $\pi^{i}$ of $X_{V}^{i}$ and one should observe that $\pi^{i}$ is obtained by collapsing the levels of $Y_{V}$ in such a way that the probability in the cell $i\in \I_{V}$ is not affected by this operation; formally, $\varpi_{i}=\pi^{i}_{V}$. For this reason, we say that $\pi^{i}$ is constructed by collapsing the levels of $Y_{V}$ \lq\lq around\rq\rq{} the cell $i$, and this trivially implies that the value of $\mu^{i}$, as well as that of $\gamma^{i}$, is unaffected by collapsing operations that do not involve the level $i_{v}$ for every $v\in V$. As a consequence of this invariance property of the \LML\ parameterization, certain zero entries in the \LML\ parameter $\gamma$ of $Y_{V}$ allow one to identify marginal independencies concerning dichotomized versions of the variables, as well as to identify levels of the variables that can be collapsed without affecting the structure of the associated bi-directed graph. We formally approach this issue by introducing  the concept of binary expansion of a discrete variable, that is based on the dichotomization (\ref{EQN:def.dichotomization}).
\begin{defn}
For $v\in V$, the \emph{binary expansion of $Y_{v}$ with respect to $J_{v}$} is the $|J_{v}|$-dimensional random vector of binary variables $X_{J_{v}}=(X_{v}^{j_{v}})_{j_{v}\in J_{v}}$.
\end{defn}
The binary expansion $X_{J_{v}}$ provides an alternative representation of $Y_{v}$: every entry of $X_{J_{v}}$ corresponds to a different dichotomization of $Y_{v}$ so that, for every $j_{v}\in J_{v}$, the variable $X_{v}^{j_{v}}$ takes on the value 1 if and only if $Y_{v}=j_{v}$ and value 0 otherwise. Moreover, $X_{J_{v}}=0_{J_{v}}$ if and only if  $Y_{v}=0_{v}$. Clearly, there exist $|I_{v}|$ different binary expansions of $Y_{v}$, depending on the choice of the baseline level, and they are all equivalent, in the sense that there is a one-to-one relationship between $Y_{v}$ and any of its binary expansions. On the other hand, the specification of one binary expansion, out of the $|I_{v}|$ existing alternatives, amounts to fixing a particular perspective from which the variable structure is explored. A suitable choice of the baseline level  may correspond to a binary expansion of special interest and, ultimately, make it possible to disclose relevant additional features concerning the association structure of the variables.
\begin{exmp}[\emph{Genetic association analysis cont.}]
Let $Y_{v}$ be the variable representing a given SNP under the codominant genotype model so that $Y_{v}$ takes values in the set $\{WM, WW, MM\}$. There exist three different binary expansions of $Y_{v}$. However, if one sets the genotype $WM$ as baseline level, then every entry of the resulting binary expansion $X_{J_{v}}$ has a clear interpretation. Indeed, one of the two entries is associated with $WW$ and therefore it corresponds to the \lq\lq dominant\rq\rq{} dichotomization $WM+MM$ vs. $WW$ whereas the other entry is associated with $MM$ and corresponds to the \lq\lq recessive\rq\rq{} dichotomization $WM+WW$ vs. $MM$. In order to make our notation more intuitive, hereafter for a SNP $Y_{v}$ we label $WW$ as $D_{v}$ (for Dominant), and $MM$ as $R_{v}$ (for Recessive) so that $J_{v}=\{D_{v}, R_{v}\}$ and $X_{J_{v}}=(X^{D_{v}}_{v}, X^{R_{v}}_{v})^{\T}$.

We now turn to the variable relative to the phenotype or trait of interest. It is common for a genetic association study to be based on a case-control design where the phenotype $Y_{v}$ is a discrete variable with levels corresponding to the controls and to the different states of a given disease for the cases. Hence, by setting the controls as baseline level, every entry of the resulting binary expansion $X_{J_{v}}$ corresponds to one of the different states of the disease.
\end{exmp}

The concept of binary expansion of a variable can then be extended to that of binary expansion of a random vector $Y_{B}$, $B\subseteq V$, with respect to $J_{B}=\cup_{v\in B} J_{v}$ that is given by $X_{J_{B}}=(X_{v}^{j_{v}})_{j_{v}\in J_{v}, v\in B}$. In the rest of the paper, we assume, without loss of generality, that $J_{B}$ is fixed and shortly write that $X_{J_{B}}$ is the binary expansion of $Y_{B}$. Furthermore, if $P=V\backslash B$ then we write $Y^{B}_{V}=(Y_{P}, X_{J_{B}})$ and say that $Y^{B}_{V}$ is the \emph{$B$-expansion} of $Y_{V}$; note that $Y_{V}^{V}=X_{J_{V}}$ whereas, if $B=\emptyset$ then $J_{B}=\emptyset$ and $Y_{V}^{B}=Y_{V}$.

The main result of this section is a generalization of Theorem~\ref{THM:gamma.eq.zero.and.disc.MP.varpi} where we show that, for every $B\subseteq V$,  bi-directed graph models for $Y^{B}_{V}$ can be parameterized by setting certain entries of $\gamma$ to zero. We emphasize that here $\gamma$ is the \LML\ parameter of $Y_{V}$ and therefore the subset $B$ plays no role in its computation.

The marginal independence of $Y^{B}_{V}=(Y_{P}, X_{J_{B}})$ is encoded by a graph on $P\cup J_{B}$ vertices. By construction, for every $v\in V$, no marginal independence is present between the entries of $X_{J_{v}}$ because $X^{j_{v}}_{v}=1$ implies that $X_{J_{v}\backslash \{j_{v}\}}=0_{J_{v}\backslash \{j_{v}\}}$, for every $j_{v}\in J_{v}$. Hence, the bi-directed graph of $Y^{B}$ is a $B$-expanded graph, formally defined below.
\begin{defn}
Let $Y_{V}$ be a discrete random variable and $V=P\dcup B$ a partition of $V$. We say that $\G^{B}$ is a \emph{$B$-expanded graph for $Y_{V}$} if it is a bi-directed graph with vertex set equal to $P\cup J_{B}$, and such that the subgraph induced by $J_{v}$ is complete for every $v\in B$.
\end{defn}

\begin{exmp}[\emph{Genetic association analysis cont.}]
Consider a case-control genetic association study where $Y_{1}$ is a variable
whose baseline level corresponds to the controls and $J_{1}=\{I, II, III\}$ encodes three different stages of a disease measured on the cases. Furthermore, let $Y_{2}$ be a SNP and $J_{2}=\{D_{2}, R_{2}\}$. The graphs in Figure~\ref{FIG:bi-dir.and.exp.1} represent: (a) the complete $B$-expanded graph with $B=\emptyset$, that is, the complete bi-directed graph with vertex set $\{1, 2\}$, (b) the complete $B$-expanded graph with $B=\{2\}$ and (c) the complete $B$-expanded graph with $B=\{1, 2\}$. Note that, to improve readability, the gray color is used to represent the complete subgraphs of the expanded variables.
\begin{figure}
 \centering
\begin{figurepdf}
 \subfigure[]
   {\includegraphics[scale=1]{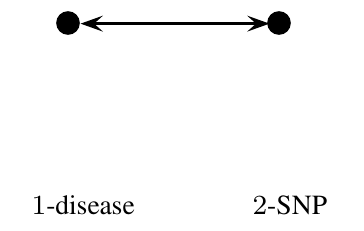}}
 \hspace{12mm}
 \subfigure[]
   {\includegraphics[scale=1]{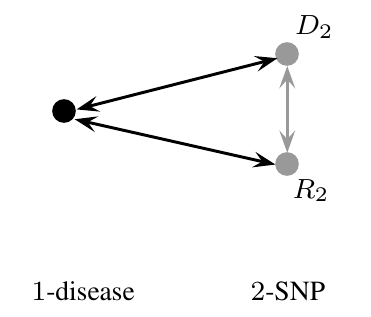}}
 \hspace{12mm}
 \subfigure[]
   {\includegraphics[scale=1]{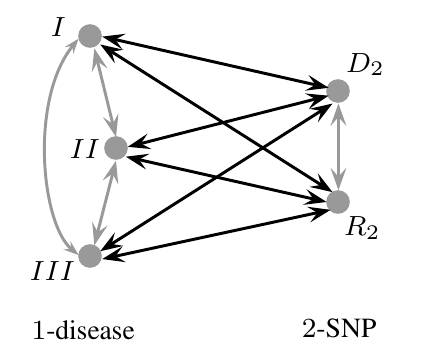}}
\end{figurepdf}
\caption{Example of complete expanded graphs for variables $Y_{1}$ and $Y_{2}$ with  $J_{1}=\{I, II, III\}$ and $J_{2}=\{D_{2}, R_{2}\}$.}
 \label{FIG:bi-dir.and.exp.1}
 \end{figure}
\end{exmp}

Before  stating the connected set Markov property on $B$-expanded graphs it is convenient to  introduce the notion of primary subsets.
\begin{defn}
Let $J_{B}=\cup_{v\in B} J_{v}$ where $B\subseteq V$ and $P=V\backslash B$. We say that $K$ is a \emph{primary} subset of $J_{B}$ if $K\subseteq  J_{B}$ and it contains at most one level for every variable in $Y_{B}$; formally, $|K\cap J_{v}|\leq 1$ for every $v\in B$. Furthermore, we say that $L$ is a primary subset of $P\cup J_{B}$ if $L\subseteq P\cup J_{B}$ and $K=L\cap J_{B}$ is a primary subset of $J_{B}$.
\end{defn}
Note that the empty set is always primary and, furthermore, if $B=\emptyset$ then $P=V$ and, in this case, every subset of $P$ is primary.
\begin{lem}\label{THM:primary.sets}
Every primary subset $L$ of $P\cup J_{B}$ can be partitioned as $L=Q\dcup K$ where $Q=P\cap L$ and $K=J_{B}\cap L$. Moreover, there exists a unique subset $D\subseteq B$ such that $K\in \J_{D}$ and, for this reason, we can write $K=j_{D}$ and $L=Q\cup j_{D}$. Conversely, for every  $Q\subseteq P$, $D\subseteq B$ and $j_{D}\in \J_{D}$ it holds that $Q\cup j_{D}$ is a primary subset of $P\cup J_{B}$.
\end{lem}
\begin{proof}
This is a straightforward consequence of the fact that if $L$ is primary and $K=L\cap J_{B}$ then $|K\cap J_{v}|\leq 1$ for every $v\in B$.
\end{proof}

We can now state the  main result of this section.

\begin{thm}\label{THM:markov-property-for-B-expanded}
Let $\gamma=(\gamma^{j_{U}})_{j\in \J_{V}, U\subseteq V}$ be the \LML\ parameter of $Y_{V}$ and let $\G^{B}=(P\cup{}J_{B}, E^{B})$ be a $B$-expanded graph for $Y_{V}$. The distribution of $Y^{B}=(Y_{P}, X_{J_{B}})$ satisfies the connected set Markov property with respect to $\G^{B}$ if and only if for every set $L\subseteq P\cup{}J_{B}$ such that
\begin{itemize}
\item[(i)] $L$  is disconnected in $\G^{B}$,
\item[(ii)] $L$ is a primary subset of $P\cup{}J_{B}$,
\end{itemize}
it holds that $\gamma^{j_{Q\cup D}}=0$ for every $j_{Q}\in \J_{Q}$, where $L=Q\cup j_{D}$ as in Lemma~\ref{THM:primary.sets} and $j_{Q\cup D}=j_{Q}\cup j_{D}$.
\end{thm}
\begin{proof}
See the Appendix~\ref{APP:00B}.
\end{proof}
The following example clarifies the connection between \LML\ parameters and edges of
expanded graphs in the simple case where only two variables are considered. Interestingly,  when both variables are expanded, every edge of the expanded graph can be associated with exactly one \LML\ parameter.
\begin{exmp}[\emph{Genetic association analysis cont.}] Let $Y_{1}$ and $Y_{2}$ be the two variable depicted in the graphs of Figure~\ref{FIG:bi-dir.and.exp.1}. The \LML\ parameter of $Y_{\{1,2\}}$ is made up of: $\gamma^{\emptyset}=0$; the main effects of $Y_{1}$ that are
$\gamma^{\{I\}}$,
$\gamma^{\{II\}}$ and
$\gamma^{\{III\}}$; the main effects of  $Y_{2}$ given by
$\gamma^{\{D_{2}\}}$ and
$\gamma^{\{R_{2}\}}$ and, finally, the two-way interactions
$\gamma^{\{I,D_{2}\}}$,
$\gamma^{\{II,D_{2}\}}$,
$\gamma^{\{III,D_{2}\}}$,
$\gamma^{\{I,R_{2}\}}$,
$\gamma^{\{II,R_{2}\}}$ and
$\gamma^{\{III,R_{2}\}}$.
It follows from Corollary~\ref{THM:gamma.eq.zero.and.disc.MP.varpi} that the edge $1\leftrightarrow 2$ in the graph (a) of Figure~\ref{FIG:bi-dir.and.exp.1} can be removed if and only if all of the six two-way \LML\ interactions are equal to zero, so that $Y_{1}\ci Y_{2}$. Consider now the graph (b) in the same figure. Here, there are only two primary subsets involving more than one vertex, $\{1, D_{2}\}$ and $\{1, R_{2}\}$. The edge $1\leftrightarrow D_{2}$ encodes the marginal association between $Y_{1}$ and the dominant version $X^{D_{2}}_{2}$ of the SNP $Y_{2}$ and, by Theorem~\ref{THM:markov-property-for-B-expanded}, it can be removed if and only if $\gamma^{\{I,D_{2}\}}=\gamma^{\{II,D_{2}\}}=\gamma^{\{III,D_{2}\}}=0$. Similarly, the edge $1\leftrightarrow R_{2}$ can be removed if and only if $\gamma^{\{I,R_{2}\}}=\gamma^{\{II,R_{2}\}}=\gamma^{\{III,R_{2}\}}=0$. It is therefore clear that $Y_{1}$ is independent of the codominant representation  $Y_{2}$ of the SNP if and only if it is independent both of the dominant representation $X^{D_{2}}_{2}$ and the recessive representation $X^{R_{2}}_{2}$. Of main interest is the case where only one edge is missing in the graph (b) because the $B$-expanded provides additional insight into the independence structure of the two variables with respect to the traditional graph for $Y_{\{1,2\}}$. We now turn to the $B$-expanded graph (c) in Figure~\ref{FIG:bi-dir.and.exp.1}. In this case, $P=\emptyset$, $J_{B}=\{I, II, III, D_{2}, R_{2}\}$ and the primary subsets of $J_{B}$ which involve more than one vertex are $\{I,D_{2}\}$, $\{II,D_{2}\}$, $\{III,D_{2}\}$, $\{I,R_{2}\}$, $\{II,R_{2}\}$ and $\{III,R_{2}\}$. Hence, it follows from  Theorem~\ref{THM:markov-property-for-B-expanded} that every edge of the graph can be removed if and only if the corresponding two way interaction is equal to zero. For instance, the edge $II\leftrightarrow R_{2}$ can be removed if and only if $\gamma^{\{II,R_{2}\}}=0$.
\end{exmp}
We now illustrate the potentiality of $B$-expanded graphs to provide interpretable parsimonious bi-directed graph submodels.
\begin{exmp}[\emph{Genetic association analysis cont.}]\label{EXA:01}
Let $V=\{1,2,3\}$ where $Y_{1}$ and $Y_{2}$ are the two variable in the graphs in Example~\ref{EXA:01} and $Y_{3}$ is an additional SNP. In this case, apart from $\gamma^{\emptyset}=0$, the \LML\ parameter is made up of 35 entries, concretely: 7 main effects, 16 two-way interactions and 12 three-way interactions. Assume that the probability distribution of $Y_{V}$ is such that
the following 18 \LML\ parameters are equal to zero:
$\gamma^{\{j_{1},R_{2}\}}$,
$\gamma^{\{j_{1},R_{3}\}}$,
$\gamma^{\{D_{2},D_{3}\}}$,
$\gamma^{\{R_{2},D_{3}\}}$,
$\gamma^{\{D_{2},R_{3}\}}$,
$\gamma^{\{j_{1},R_{2},R_{3}\}}$,
$\gamma^{\{j_{1},R_{2},D_{3}\}}$ and
$\gamma^{\{j_{1},D_{2},R_{3}\}}$ for every $j_{1}\in J_{1}$. It is easy to see from Theorem~\ref{THM:gamma.eq.zero.and.fact.mu} that in this case there are no pairwise marginal independencies so that the distribution of $Y_{V}$ is associated with the complete graph (a) in Figure~\ref{FIG:bi-directe.and.expanded}. However, the distribution of $Y_{V}$ belongs to a parsimonious model that can be completely defined in terms of marginal independence relationships. If we set $B=\{2, 3\}$, then  the zero \LML\ parameters above are associated with the subsets of $\{1\}\cup J_{B}$ with which are both primary and disconnected in the graph (b) of Figure~\ref{FIG:bi-directe.and.expanded}. Hence, by Theorem~\ref{THM:markov-property-for-B-expanded}, the distribution of $Y^{B}_{V}$ satisfies the connected set Markov property with respect to the latter graph that implies, among others, the independence of the disease of the recessive versions of the SNPs; $Y_{1}\ci (X_{2}^{R_{2}}, X_{3}^{R_{3}})$. Clearly, it would be possible also to expand $Y_{1}$, but this would make the graph unnecessarily more complex because the zero structure of $\gamma$ does not allow us to remove any edge with an endpoint in the expansion of this variable. However, if additional interactions are equal to zero such as, for instance, $\gamma^{\{III,D_{2}\}}$
$\gamma^{\{II,D_{3}\}}$,
$\gamma^{\{III,D_{3}\}}$,
$\gamma^{\{III,D_{2},D_{3}\}}$ and
$\gamma^{\{II,D_{2},D_{3}\}}$, then it makes sense to expand also $Y_{1}$ so as to obtain the $B$-expanded graph (c) in Figure~\ref{FIG:bi-directe.and.expanded}. This graph encodes marginal relationships involving single levels of $Y_{1}$, such as $X_{1}^{III}\ci (X_{2}^{R_{2}}, X_{2}^{D_{2}}, X_{3}^{R_{3}}, X_{3}^{D_{3}})$ that is equivalent to $X_{1}^{III}\ci Y_{\{2, 3\}}$.
Note that, in the latter model, 23 out of the 35 \LML\ parameters are constrained to zero but the independence structure of $Y_{V}$ is still represented by the complete graph.
\end{exmp}
\begin{figure}
 \centering
\begin{figurepdf}
 \subfigure[]
   {\includegraphics[scale=1]{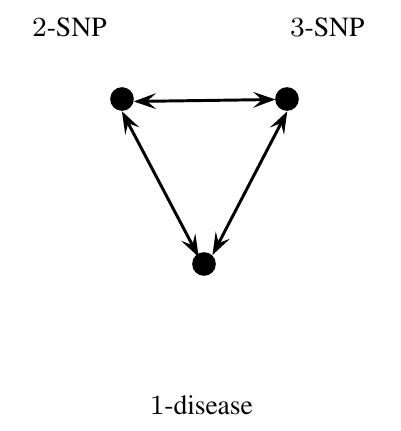}}
 \hspace{9mm}
 \subfigure[]
   {\includegraphics[scale=1]{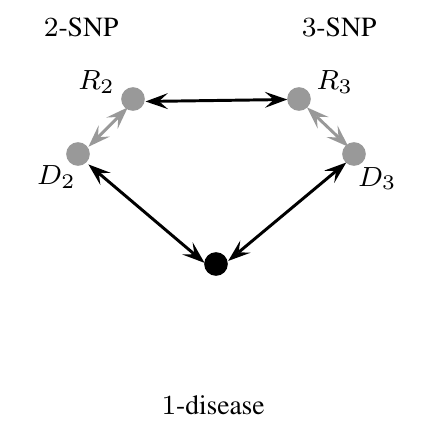}}
 \hspace{9mm}
 \subfigure[]
   {\includegraphics[scale=1]{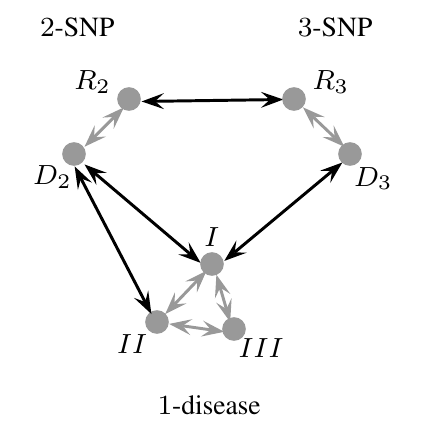}}
\end{figurepdf}
 \caption{Example of expanded graphs for $Y_{\{1,2,3\}}$ with  $J_{1}=\{I, II, III\}$ and $J_{i}=\{D_{i}, R_{i}\}$ for $i=2,3$.}\label{FIG:bi-directe.and.expanded}
 \end{figure}

The rest of this section is devoted to some basic results. These are required for the proof of Theorem~\ref{THM:markov-property-for-B-expanded} but are also of interest on their own because they provide a formal proof of the dichotomization invariance property of both the \Mob\ and the \LML\ parameterizations. It is worth remarking that caution needs to be used in the application of the existing results for marginal independence models to $Y^{B}_{V}$ because its distribution contains structural zeros. We exploit the well known fact that a random variable $Z$ is independent of $Y_{v}$ if and only if $Z$ is independent of  $X_{v}^{j_{v}}$  for every $j_{v}\in J_{v}$. This is formally stated below in a slightly different form.
\begin{lem}\label{THM:lemma.basic.indep}
Let $X_{J_{v}}=(X_{v}^{j_{v}})_{j_{v}\in J_{v}}$ the binary expansion of $Y_{v}$ and let $Z$ be a discrete random vector. For a subvector $J^{\prime}_{v}\subset J_{v}$ and an element $j_{v}\in J_{v}\backslash J^{\prime}_{v}$ it holds that $(X_{J^{\prime}_{v}}, X_{v}^{j_{v}})\ci Z$ if and only if both $X_{J^{\prime}_{v}}\ci Z$ and $X_{v}^{j_{v}}\ci Z$.
\end{lem}
The lemma above can then be directly applied to show that the Markov property for $B$-expanded graphs is characterized by the same property for the collection of subvectors $(Y_{P}, X_{j_{B}})$, for $j_{B}\in \J_{B}$, whose probability distributions have no structural zeros.

\begin{lem}\label{EQN:lemma.expanded.Markov}
Let $\G^{B}=(P\cup J_{B}, E^{B})$ be a $B$-expanded graph for $Y_{V}$. The distribution of $Y^{B}=(Y_{P}, X_{J_{B}})$ satisfies the connected set Markov property with respect to $\G^{B}$ if and only if the distribution of $(Y_{P}, X_{j_{B}})$ satisfies the connected set Markov property with respect to $\G^{B}_{P\cup j_{B}}$ for every $j_{B}\in \J_{B}$.
\end{lem}

The probability table of $Y_{B}$ is made up of $|\I_{B}|$ positive probabilities that can be written as $\pr(Y_{B}=i_{B})$ for every $i_{B}\in \I_{B}$ or, equivalently, as $\pr(Y_{D}=j_{D}, Y_{B\backslash D}=0_{B\backslash D})$ for every $j_{D}\in \J_{D}$ and $D\subseteq B$. We have named $X_{J_{B}}$ the binary expansion of $Y_{B}$ because its probability table can be hypothetically constructed by adding $2^{|J_{B}|}-|\I_{B}|$ structural-zero cells to the probability table of $Y_{B}$. More specifically, the probability distribution of $X_{J_{B}}$ can be written as $\pr(X_{K}=1_{K}, X_{J_{B}\backslash K}=0_{J_{B}\backslash K})$, for $K\subseteq J_{B}$ and, these probabilities are equal to zero if and only if $K$ is non-primary.
\begin{lem}\label{THM:zero.probability}
Let $X_{J_{B}}$ be the binary expansion of $Y_{B}$ and let $K\subseteq J_{B}$. Furthermore, let $\tilde{\mu}=(\tilde{\mu}_{K})_{K\subseteq J_{B}}$ and $\tilde{\gamma}=(\tilde{\gamma}_{K})_{K\subseteq J_{B}}$ be the \Mob\ and the \LML\ parameter of $X_{J_{B}}$, respectively, whereas $\mu$ and $\gamma$ are  the \Mob\ and the \LML\ parameter of $Y_{V}$, respectively. If $K$ is a primary subset of $J_{B}$ then we can write $K=j_{D}$ by Lemma~\ref{THM:primary.sets} and it holds both that
$\pr(X_{K}=1_{K}, X_{J_{B}\backslash K}=0_{J_{B}\backslash K})
=\pr(Y_{D}=j_{D}, Y_{B\backslash D}=0_{B\backslash D})$
and, if $j_{D}\neq\emptyset$, that $\pr(X_{K}=1_{K})=\pr(Y_{D}=j_{D})$. Furthermore,  it holds that both $\mu^{j_{D}}=\tilde{\mu}_{j_{D}}$ and $\gamma^{j_{D}}=\tilde{\gamma}_{j_{D}}$. On the other hand, if $K$ is non-primary then $\pr(X_{K}=1_{K}, X_{J_{B}\backslash K}=0_{J_{B}\backslash K})=\pr(X_{K}=1_{K})=0$, $\tilde{\mu}_{K}=0$ and, consequently, $\tilde{\gamma}_{K}$ is not well-defined.
\end{lem}
\begin{proof}
See the Appendix~\ref{APP:00B}.
\end{proof}

We turn now to $Y^{B}_{V}=(Y_{P}, X_{J_{B}})$. The \Mob\ and the \LML\ parameters of $Y_{V}$ are
$\mu=(\mu^{j_{U}})_{j\in \J_{V}, U\subseteq V}$ and $\gamma=(\gamma^{j_{U}})_{j\in \J_{V}, U\subseteq V}$, respectively. Furthermore, we denote the \Mob\ and the \LML\ parameters of $Y^{B}_{V}$ by $\tilde{\mu}$ and $\tilde{\gamma}$, respectively, where the entries of $\tilde{\mu}$ are  $\tilde{\mu}^{j_{Q}\cup 1_{K}}$ for every $j_{Q}\in \J_{Q}$ and $K\subseteq J_{B}$, and we use the shorthand $\tilde{\mu}^{j_{Q}\cup 1_{K}}=\tilde{\mu}^{j_{Q}}_{K}$.
Similarly, we denote the entries of $\tilde{\gamma}$ by $\tilde{\gamma}_{K}^{j_{Q}}$ for  $j_{Q}\in \J_{Q}$ and  $K\subseteq J_{B}$. The following theorem states that both the \Mob\ and the \LML\ parameters are dichotomization invariant in the sense that $\mu$ and $\gamma$ are subvectors of $\tilde{\mu}$ and $\tilde{\gamma}$, respectively, whereas the remaining entries of $\tilde{\mu}$ and $\tilde{\gamma}$ are equal to zero and not well-defined, respectively.
\begin{thm}\label{THM:cor.to.binary.exp.prop}
Let $Y^{B}_{V}=(Y_{P}, X_{J_{B}})$ be the  $B$-expansion of $Y_{V}$.
Furthermore, let $\tilde{\mu}^{j_{Q}}_{K}$ and $\tilde{\gamma}^{j_{Q}}_{K}$, for $j_{Q}\in \J_{Q}$ and $K\subseteq J_{B}$, the \Mob\ and the \LML\ parameter of $Y^{B}$, respectively, whereas $\mu$ and $\gamma$ are  the \Mob\ and the \LML\ parameter of $Y_{V}$, respectively. If $L$ is a primary subset of $P\cup J_{B}$, so that $L=Q\cup j_{D}$ as in Lemma~\ref{THM:primary.sets}, then
$\mu^{j_{Q\cup D}}=\tilde{\mu}_{j_{D}}^{j_{Q}}$ and $\gamma^{j_{Q\cup D}}=\tilde{\gamma}_{j_{D}}^{j_{Q}}$, for every $j_{Q}\in \J_{Q}$.
Conversely, if $L$ is non-primary and we set $K=L\cap J_{B}$ and $Q=L\cap P$, then $\tilde{\mu}_{K}^{j_{Q}}=0$ and, consequently, $\tilde{\gamma}_{K}^{j_{Q}}$ is not well-defined, for every $j_{Q}\in \J_{Q}$.
\end{thm}
\begin{proof}
See the Appendix~\ref{APP:00B}.
\end{proof}
\section{Applications}\label{SEC:applications}
Learning the structure of a bi-directed graph model from data requires the specification of a strategy for the exploration of the set of candidate structures. Typically, exhaustive search is unfeasible because the cardinality of the space of all possible structures grows exponentially with the number of vertices of the graph. More specifically, there are  $2^{p \choose 2}$ bi-directed graphs on $|V|=p$ vertices, but the dimension of the search space further increases when expanded graphs are considered. For instance, if every variable has exactly 3 levels, then the complete $V$-expanded graph has 4 edges joining every pair of variables and the number of candidate structures grows to $\left\{2^{p \choose 2}\right\}^{4}$.
For the analyses in this section, we applied a search strategy designed to explore the relationship between expanded versions of the variables. The procedure starts from an arbitrary ordering of the variables, $Y_{(1)},\ldots, Y_{(p)}$. Then, for the pair $(Y_{(1)}, Y_{(2)})$ the graph  $G^{V}$ is set equal to the complete bi-directed graph and the family of graphs given by $G^{V}$ and all its subgraphs obtained by removing edges between the expanded versions of $Y_{(1)}$ and $Y_{(2)}$ is considered. An exhaustive search is performed within this set of candidate structures and a model is selected on the basis of a predefined criterion. A similar exhaustive search is then performed for every pair of variables, from $(Y_{(1)}, Y_{(3)})$ to $(Y_{(p-1)}, Y_{(p)})$, with the only modification that, at every step, $G^{V}$ is set equal to the graph of the model selected in the previous step. Clearly, the model resulting from the application of this procedure may depend on the initial ordering of the variable and, to avoid this arbitrariness, the search procedure above is run $p!$ times, once for every permutation of the variables. Hence, one of the resulting $p!$, not necessarily distinct, models is selected by further application of the predefined criterion.

The criterion we use here to select a model within a set of candidate models is to choose the model with the optimal value of the Bayesian information criterion (BIC) among those whose $p$-value, computed on the basis of the asymptotic chi-squared distribution of the deviance, is no smaller than $0.05$.

Maximum likelihood estimation for \LML\ models under a multinomial or Poisson sampling scheme is a constrained optimization problem, that can be carried out by using gradient-based ascent methods; we refer to \citet{lang:1996} and \citet{ber-al:2009} for details. We remark that the likelihood function can be expressed in terms both of \Mob\ parameters and of \LML\ parameters but not analytically in terms of multivariate logistic parameters because an analytic form of the inverse map to compute $\varpi$ from the marginal logistic parameters is not available; see also \citet{roverato2013log}.

\subsection{Satisfaction with housing data}\label{SEC:app.satisfaction}
\citet[][exercise 8.28]{agresti2013categorical} gives a $4\times 3\times 3\times 2$ table originally described by \citet{madsen1976statistical} that refers to a sample of 1681 residents of Copenhagen. The variables are $Y_{1}=\;$\emph{type of housing} (levels: $Ap=\;$\emph{apartments}, $At=\;$\emph{atrium houses}, $Te=\;$\emph{terraced houses}, $To=\;$\emph{tower blocks}), $Y_{2}=\;$\emph{feeling of influence on apartment management} (levels: $L_{2}=\;$\emph{low}, $M_{2}=\;$\emph{medium}, $H_{2}=\;$\emph{high}), $Y_{3}=\;$\emph{satisfaction with housing conditions} (levels:  $L_{3}=\;$\emph{low}, $M_{3}=\;$\emph{medium}, $H_{3}=\;$\emph{high}) and  $Y_{4}=\;$\emph{degree of contact with other residents} (levels: $L_{4}=\;$\emph{low}, $H_{4}=\;$\emph{high}).
An analysis of these data within the family of undirected graphical models did not highlight any structural relationship between variables because no pairwise conditional independence could be identified on the basis of the asymptotic chi-squared distribution of the deviance. However, an exhaustive search within the family of bi-directed graph models resulted in the graph in Figure~\ref{FIG:housing-data-A} which encodes the structural relationship  $Y_{3}\ci Y_{4}$, that is the marginal independence of the satisfaction with housing condition from the degree of contact with other residents. This model has deviance $5.13$ on 2 degrees of freedom ($p=0.08$, BIC=$-9.7$).
\begin{figure}
 \centering
\begin{figurepdf}
 \subfigure[][\label{FIG:housing-data-A}]
   {\includegraphics[scale=1]{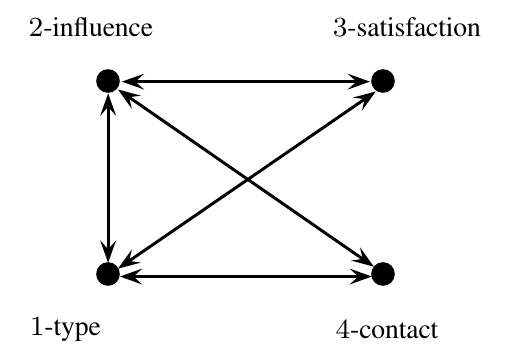}}
 \hfill
 \subfigure[][\label{FIG:housing-data-B}]
   {\includegraphics[scale=1]{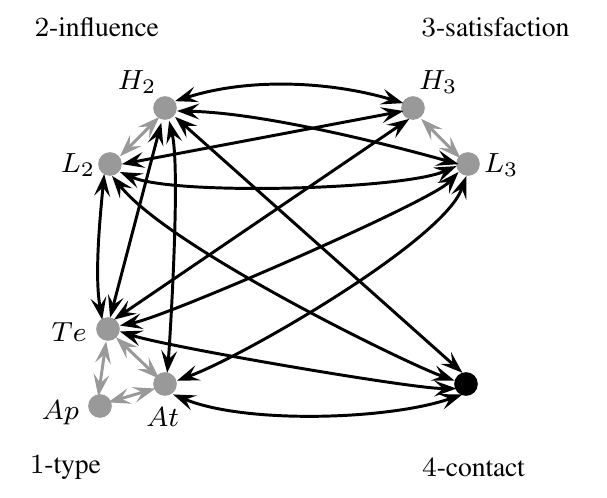}}
 \subfigure[][\label{FIG:housing-data-C}]
   {\includegraphics[scale=1]{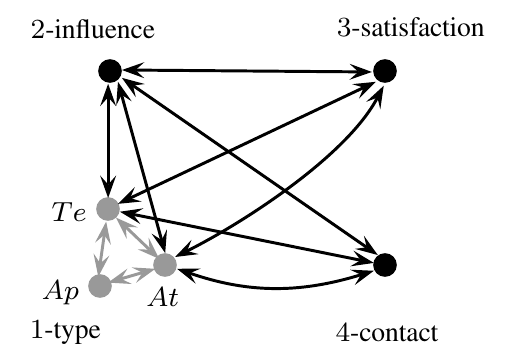}}
\end{figurepdf}
 \caption{Housing data: three alternative graphical representations of the selected model.}\label{FIG:housing-data}
 \end{figure}
We then considered the expanded version of the variables. Variable $Y_{4}$ is binary and, therefore, it can not be expanded. Variables $Y_{2}$ and $Y_{3}$ have 3 ordinal levels and the central level \emph{medium} is a natural baseline because it makes sense to collapse it with either the level \emph{low} or the level \emph{high} leading to the binary expansions $(X_{v}^{L_{v}}, X_{v}^{H_{v}})^{\T}$ for $v=2,3$. The variable $Y_{1}$ has no natural baseline and we arbitrarily set the level \emph{tower blocks} as baseline. We then applied the structural learning procedure described above thereby identifying the model represented in Figure~\ref{FIG:housing-data-B} with  deviance $34.34$ on 23 degrees of freedom ($p=0.06$, BIC=$-136.5$). This is a parsimonious submodel of the model in  Figure~\ref{FIG:housing-data-A} with 21 additional constraints that can be interpreted as (i) $X_{1}^{Ap}\ci (Y_{2}, Y_{3},Y_{4})$ and (ii) $X_{1}^{At}\ci (X_{2}^{L_{2}}, X_{3}^{H_{3}})$. The graph in Figure~\ref{FIG:housing-data-B} fully identifies the selected model, but has the disadvantage that it is dense and, to some extent, difficult to read. For this reason, we believe that the graph in Figure~\ref{FIG:housing-data-C} where only variable $Y_{1}$ is expanded, although failing to represent the independence $X_{1}^{At}\ci (X_{2}^{L_{2}}, X_{3}^{H_{3}})$, may provide a good compromise between readability  of the graphical representation and completeness of the information provided by the graph.

\subsection{SNPs data}
The package \texttt{SNPassoc} \citep{Gonzalez2012SNPassoc} for the statistical sofware R \citep{R2013} supplies data from the HapMap project (\url{http://www.hapmap.org}) with 120 observations on the 9808 SNPs and a phenotype binary variable, named \emph{group} indicating the population (European population and Yoruba). The genomic information concerning names of SNPs, chromosomes and genetic position is also available. We applied the function \texttt{getSignificantSNPs} provided in the package, with the default parameter values, to extract the 21 SNPs in chromosome number one that are significantly associated with the phenotype under the codominant genotype model. This first step of the analysis may be followed by more specific analysis aimed to understand the association structure between significant SNPs. However, carrying out a fully multivariate analysis represents a challenging task because table of counts are highly sparse and, in fact, since low  frequencies are naturally associated with the mutant alleles, even three way marginal tables have typically a large proportion of zero entries. For this reason, models of marginal independence may prove useful because the association structures learnt by considering subsets of variables may be used to hypothesize the joint behaviour of a larger sets of variables. As an example, we denoted the 21 significant SNPs by SNP-$i$ for $i=1,\ldots, 21$, according to their position on the chromosome, and investigated the association structure between neighbouring SNPs. More specifically, for every $i=1,\ldots, 20$  we applied an exhaustive search to learn an expanded bi-directed graph between SNP-$i$ and its neighbour SNP-$(i+1)$. In most cases the saturated model was selected, but some interesting feature emerged. For instance, according to the selected models, SNP-9 (named rs11165510) is associated with its neighbours SNP-8 (named rs10782748) and SNP-10 (named rs1413241) only through its dominant version as represented in Figure~\ref{FIG:SNP-data-A}, with deviances $4.23$ and $1.60$, respectively, on 2 degrees of freedom ($p=0.12$, BIC=$-5.34$  and $p=0.44$, BIC=$-7.91$, respectively). An exhaustive search for the association between SNP-8 and SNP-10 produced the saturated model and the union of the three selected models gives the model represented by the graph in Figure~\ref{FIG:SNP-data-A} that provides  an adequate fit of the data with  deviance $8.85$ on 8 degrees of freedom ($p=0.35$, BIC=$-29.24$).
\begin{figure}
 \centering
\begin{figurepdf}
 \subfigure[][\label{FIG:SNP-data-A}]
   {\includegraphics[scale=1]{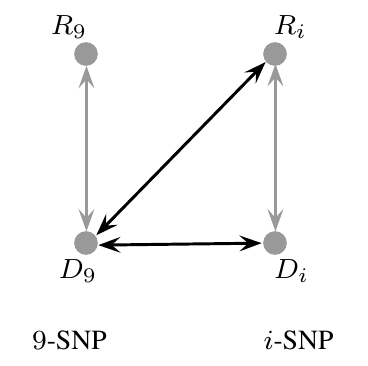}}
 \hspace{3cm}
 \subfigure[][\label{FIG:SNP-data-B}]
   {\includegraphics[scale=1]{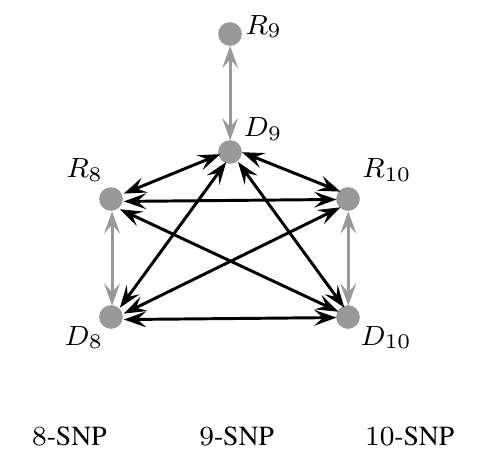}}
\end{figurepdf}
\caption{SNPs data: expanded bi-directed graphs for the selected models of (a) SNP-9 and SNP-$i$ for $i=8,10$ and (b)  SNP-8, SNP-9 and SNP-10.}
 \label{FIG:SNP-data}
\end{figure}

\section{Discussion}\label{SEC:discussion}

\citet{roverato2013log} introduced the \LML\ parameterization for binary data and explained its advantages with respect to both the multivariate logistic and the \Mob\ parameterization. The extension of this approach to variables with an arbitrary number of levels, has disclosed the additional invariance property that makes it possible to use the \LML\ parameter of $Y_{V}$ to characterize the connected set Markov property for any $B$-expansion $Y^{B}_{V}$ of $Y_{V}$. The idea of using
the indicator variables (\ref{EQN:def.dichotomization}) in the treatment of categorical variables with more than two levels is not new; see among others \citet[][Section~7]{ravikumar2010high} and \citet[][Section~13.3.2]{tutz2011regression}. However, this approach seems especially powerful in the treatment of marginal models and, furthermore, the \LML\ parameterization allows one to fully exploit its advantages. Indeed, using the \LML\ parameterization amounts to implicitly working with the binary expansions of variables or, from a different perspective, the \LML\ parameterization allows one to deal with expanded variables implicitly, so as to avoid all the difficulties associated with the expansion operation. For instance, one has not to worry either for the presence of structural zeros or for the inefficiencies deriving from the artificial increase in dimensionality. In structural learning, the set $B$ does not need to be defined a priori, but it can be specified after a \LML\ model has been selected from data. For instance, $B$ can be chosen so as to optimize the trade-off between readability of the graph and the need to explicit the learnt independencies involving every single expanded variable, as shown in the application of Section~\ref{SEC:app.satisfaction}.

Open questions include the specification of the baseline level in situations where there is no \lq{}natural\rq{} baseline level, as well as the specification of the baseline level for ordinal variables. In Section~\ref{SEC:applications} we implemented a naive structural learning procedure, but the general issue of developing model search strategies for the identification of a parsimonious bi-directed graph submodel is still open. We note that it may be appealing to restrict the search space to the models characterized by the family of $V$-expanded graphs, because it is made up of interpretable models, and it is smaller than the the family of all the models obtained by constraining a subset of the \LML\ parameters to zero.

We close this discussion by remarking that, by Theorem~\ref{THM:cor.to.binary.exp.prop}, also the \Mob\ parameterization satisfies the same dichotomization invariance property as the \LML\ parameterization. However, the \LML\ parameterization has the advantage that marginal independence relationships correspond to, linear, zero constraints in the space of the parameters.

\section*{Acknowledgments}
We gratefully acknowledge useful discussions with Robert Castelo, Monia Lupparelli and Luca la Rocca.

\appendix

\section{Proofs}\label{APP:00B}
\subsection*{Proof of Theorem~\ref{THM:poly.zero.gamma.and.fact.mu}}
\noindent{}(i)$\,\Rightarrow\,$(ii). Firstly, we note that the factorizations in (ii) are trivially true when at least one between $A^{\prime}$ and $B^{\prime}$ is equal to the empty set because $\mu^{j_{\emptyset}}=1$. The independence $Y_{A}\ci Y_{B}$ implies that $Y_{A^{\prime}}\ci Y_{B^{\prime}}$ for every $A^{\prime}\subseteq A$ and $B^{\prime}\subseteq B$ with $A^{\prime},B^{\prime}\neq \emptyset$. In turn, $Y_{A^{\prime}}\ci Y_{B^{\prime}}$ implies that, for every $j\in \J_{V}$,  $\pr(Y_{A^{\prime}}=j_{A^{\prime}}, Y_{B^{\prime}}=j_{B^{\prime}})=\pr(Y_{A^{\prime}}=j_{A^{\prime}})\times \pr(Y_{B^{\prime}}=j_{B^{\prime}})$ and since $Y_{A^{\prime}}=j_{A^{\prime}}$ if and only if $X^{j}_{A^{\prime}}=1_{A^{\prime}}$, and similarly for $Y_{B^{\prime}}$, this implies that $\pr(X^{j}_{A^{\prime}}=1_{A^{\prime}}, X^{j}_{B^{\prime}}=1_{B^{\prime}})=\pr(X^{j}_{A^{\prime}}=1_{A^{\prime}})\times \pr(X^{j}_{B^{\prime}}=1_{B^{\prime}})$ and the result follows by (\ref{EQN:mu.content}) because the latter factorization can be written as  $\mu^{j_{A^{\prime}\cup B^{\prime}}}=\mu^{j_{A^{\prime}}}\times \mu^{j_{B^{\prime}}}$.

\noindent{}(i)$\Leftarrow$(ii). By Theorem~\ref{THM:gamma.eq.zero.and.fact.mu}, (ii) implies that $X_{A}^{j}\ci X_{B}^{j}$ for every $j\in \J_{V}$, and therefore that
$\pr(X^{j}_{A}=1_{A}, X^{j}_{B}=1_{B})=\pr(X^{j}_{A}=1_{A})\times \pr(X^{j}_{B}=1_{B})$, for every $j\in {\mathcal J}_{V}$. The latter factorization is equivalent to
\begin{eqnarray}\label{EQN:001}
\pr(Y_{A}=j_{A}, Y_{B}=j_{B})=\pr(Y_{A}=j_{A})\times \pr(Y_{B}=j_{B})
\quad\mbox{for every }j\in {\mathcal J}_{V}
\end{eqnarray}
and we have to show that the factorization in (\ref{EQN:001}) holds for every $i\in \I_{V}$. This follows immediately form Theorem~8 of \citet{drton2009discrete} but we formally prove it for completeness. Every $i\in \I_{V}$ with $i\not\in \J_{V}$ contains at least one level labeled as baseline, that is as \lq\lq{}0\rq\rq{}, and the proof is by induction on the number baseline levels in $i$. Let $k$ denote the number of baseline levels in $i\in \I_{V}$; the factorization in (\ref{EQN:001}) holds for $k=0$ and we show that if it is true for $k=r-1$, with $r>0$, then it is also true for $k=r$. Assume that $k=r$ and that, without loss of generality, $Y_{v}=0_{v}$ with $v\in A$. Hence, if $A^{\prime}=A\backslash \{v\}$,
\begin{eqnarray*}
\pr(Y_{A}=i_{A}, Y_{B}=i_{B})
&=& \pr(Y_{v}=0_{v}, Y_{A^{\prime}}=i_{A^{\prime}}, Y_{B}=i_{B})\\
&=& \pr(Y_{A^{\prime}}=i_{A^{\prime}}, Y_{B}=i_{B})-\sum_{i_{v}=1_{v}}^{d_{v}}\pr(Y_{v}=i_{v}, Y_{A^{\prime}}=i_{A^{\prime}}, Y_{B}=i_{B})
\end{eqnarray*}
and, since the number of \lq\lq{}0\rq\rq{}'s in $i_{A^{\prime}\cup B}$ is equal to $r-1$, it follows from the induction assumption that
\begin{eqnarray*}
\pr(Y_{A}=i_{A}, Y_{B}=i_{B})
&=& \pr(Y_{A^{\prime}}=i_{A^{\prime}})\times \pr(Y_{B}=i_{B})-\sum_{i_{v}=1_{v}}^{d_{v}}\pr(Y_{v}=i_{v}, Y_{A^{\prime}}=i_{A^{\prime}})\times \pr(Y_{B}=i_{B})\\
&=&\left\{\pr(Y_{A^{\prime}}=i_{A^{\prime}})-\sum_{i_{v}=1_{v}}^{d_{v}}\pr(Y_{v}=i_{v}, Y_{A^{\prime}}=i_{A^{\prime}})\right\}\pr(Y_{B}=i_{B})\\
&=& \pr(Y_{A}=i_{A})\times \pr(Y_{B}=i_{B})
\end{eqnarray*}
as required.

\noindent{}(ii)$\,\Leftrightarrow\,$(iii). This  follows immediately from the equivalence, for every $j\in {\mathcal J}_{V}$, of (ii) and (iii) in Theorem~\ref{THM:gamma.eq.zero.and.fact.mu}.

\subsection*{Proof of Corollary~\ref{THM:gamma.eq.zero.and.disc.MP.varpi}}
As a consequence of Theorem~4 of \citet{ric:2003},  we have to show that $Y_{C_{1}}\ci\cdots\ci Y_{C_{r}}$ for every disconnected set $U$ of $\G$ if and only if for every disconnected set $U$ of $\G$ it holds that $\gamma^{j_{U}}=0$ for every $j_{U}\in \J_{U}$. Recall that $C_{1},\ldots,C_{r}$ are the connected components of $U$ and that $r\geq 2$.

If $Y_{C_{1}}\ci\cdots\ci Y_{C_{r}}$ then we can set $A=C_{1}$ and $B=C_{2}\dcup\cdots\dcup C_{r}$ so that $U=A\dcup B$ and $Y_{A}\ci Y_{B}$. The result follows by noticing that
every $j_{U}\in \J_{U}$ can be written as $j_{A\cup B}$ so that $\gamma^{j_{U}}=\gamma^{j_{A\cup B}}=0$ by Theorem~\ref{THM:poly.zero.gamma.and.fact.mu}.

We now show the reverse implication, that is, that if $\gamma^{j_{U}}=0$ for every $j_{U}\in \J_{U}$ such that $U$ is  disconnected in $\G$, then for every disconnected set $U$ of $\G$ it holds that $Y_{C_{1}}\ci\cdots\ci Y_{C_{r}}$. Let $A$ and $B$ be defined as above. Then, for every $A^{\prime}\subseteq A$ and $B^{\prime}\subseteq B$, with $A^{\prime},B^{\prime}\neq\emptyset$ the set $A^{\prime}\cup B^{\prime}$ is disconnected in $\G$ so that, by assumption, $\gamma^{j_{A^{\prime}\cup B^{\prime}}}=0$ for every $j_{A^{\prime}\cup B^{\prime}}\in \J_{A^{\prime}\cup B^{\prime}}$. This is equivalent to saying that
$\gamma^{j_{A^{\prime}\cup B^{\prime}}}=0$ for every $j\in \J_{V}$, $A^{\prime}\subseteq A$ and $B^{\prime}\subseteq B$ such that $A^{\prime}\neq\emptyset$ and $B^{\prime}\neq \emptyset$. The latter, by Theorem~\ref{THM:poly.zero.gamma.and.fact.mu}, implies that $Y_{A}\ci Y_{B}$ or, equivalently, that $Y_{C_{1}}\ci Y_{C_{2}\cup\cdots\cup C_{r}}$.
The same procedure can then be applied, for every $i=1,\ldots, r-1$, with $A=C_{i}$ and $B=C_{i+1}\cup\cdots\cup C_{r}$ to show that $Y_{C_{i}}\ci Y_{C_{i+1}\cup\cdots\cup C_{r}}$ for every $i=1,\ldots, r-1$ which, in turn, implies that $Y_{C_{1}}\ci\cdots\ci Y_{C_{r}}$, as required.

\subsection*{Proof of Lemma~\ref{THM:lemma.basic.indep}}
It is trivially true that $(X_{J^{\prime}_{v}}, X_{v}^{j_{v}})\ci Z$ implies both $X_{J^{\prime}_{v}}\ci Z$ and $X_{v}^{j_{v}}\ci Z$ and therefore it is sufficient to prove the inverse implication that is, that if both
\begin{eqnarray}\label{EQN:111-1}
\pr(X_{J^{\prime}_{v}}=
x_{J^{\prime}_{v}},Z=z)=\pr(X_{J^{\prime}_{v}}=x_{J^{\prime}_{v}})
\times \pr(Z=z)
\end{eqnarray}
and
\begin{eqnarray}\label{EQN:111-2}
\pr(X_{v}^{j_{v}}=x_{v}^{j_{v}}, Z=z)=\pr(X_{v}^{j_{v}}=x_{v}^{j_{v}})\times \pr(Z=z)
\end{eqnarray}
hold true for every $x_{J^{\prime}_{v}}$, $x_{J^{\prime}_{v}}$ and $z$ in the state pace of the corresponding variables, then it also holds that
\begin{eqnarray}\label{EQN:111-3}
\pr(X_{J^{\prime}_{v}}=x_{J^{\prime}_{v}}, X_{v}^{j_{v}}=x_{v}^{j_{v}}, Z=z)=\pr(X_{J^{\prime}_{v}}=x_{J^{\prime}_{v}}, X_{v}^{j_{v}}=x_{v}^{j_{v}})\times \pr(Z=z).
\end{eqnarray}

There are three possible cases: (i) the vector $(x_{J^{\prime}_{v}}, x_{v}^{j_{v}})$ contains more than one entry equal to 1; (ii) exactly one entry of $(x_{J^{\prime}_{v}}, x_{v}^{j_{v}})$ is equal to 1 and, (iii), all the entries  of $(x_{J^{\prime}_{v}}, x_{v}^{j_{v}})$ are equal to zero.

In case (i) the factorization (\ref{EQN:111-3}) trivially holds true because both $\pr(X_{J^{\prime}_{v}}=x_{J^{\prime}_{v}}, X_{v}^{j_{v}}=x_{v}^{j_{v}}, Z=z)=0$ and $\pr(X_{J^{\prime}_{v}}=x_{J^{\prime}_{v}}, X_{v}^{j_{v}}=x_{v}^{j_{v}})=0$. In the case (ii), for $x_{v}^{j_{v}}=1$ it holds that
$\pr(X_{J^{\prime}_{v}}=0_{J^{\prime}_{v}}, X_{v}^{j_{v}}=1, Z=z)=\pr(X_{v}^{j_{v}}=1, Z=z)$ and $\pr(X_{J^{\prime}_{v}}=0_{J^{\prime}_{v}}, X_{v}^{j_{v}}=1)=\pr(X_{v}^{j_{v}}=x_{v}^{j_{v}})$
because $X_{v}^{j_{v}}=1$ implies $X_{J^{\prime}_{v}}=0_{J^{\prime}_{v}}$, and therefore the factorization (\ref{EQN:111-3}) is implied by (\ref{EQN:111-2}). Similarly, if $x_{v}^{j_{v}}=0$ then  it holds that
$\pr(X_{J^{\prime}_{v}}=x_{J^{\prime}_{v}}, X_{v}^{j_{v}}=0, Z=z)=\pr(X_{v}^{j_{v}}=x_{v}^{j_{v}}, Z=z)$ and $\pr(X_{J^{\prime}_{v}}=x_{J^{\prime}_{v}}, X_{v}^{j_{v}}=0)=\pr(X_{v}^{j_{v}}=x_{v}^{j_{v}})$ and therefore the factorization (\ref{EQN:111-3}) is implied by (\ref{EQN:111-1}). Finally, in the case (iii) it holds that
\begin{eqnarray*}\label{EQN:111-4}
\pr(X_{J^{\prime}_{v}}=0_{J^{\prime}_{v}}, X_{v}^{j_{v}}=0, Z=z)
&=&\pr(X_{J^{\prime}_{v}}=0_{J^{\prime}_{v}}, Z=z)-\pr(X_{J^{\prime}_{v}}=0_{J^{\prime}_{v}}, X_{v}^{j_{v}}=1, Z=z)\\
&=&\pr(X_{J^{\prime}_{v}}=0_{J^{\prime}_{v}})\times \pr(Z=z)-\pr(X_{J^{\prime}_{v}}=0_{J^{\prime}_{v}}, X_{v}^{j_{v}}=1)\times \pr(Z=z)\\
&=&\left\{\pr(X_{J^{\prime}_{v}}=0_{J^{\prime}_{v}})-\pr(X_{J^{\prime}_{v}}=0_{J^{\prime}_{v}}, X_{v}^{j_{v}}=1)\right\}\times \pr(Z=z)\\
&=& \pr(X_{J^{\prime}_{v}}=0_{J^{\prime}_{v}}, X_{v}^{j_{v}}=0)\times \pr(Z=z)
\end{eqnarray*}
where the two factorizations on the right hand side of the second equality follow from (\ref{EQN:111-1}) and the case (ii) above, respectively.

\subsection*{Proof of Lemma~\ref{EQN:lemma.expanded.Markov}}

If the distribution of a set of variables satisfies the connected set Markov property with respect to a given bi-directed graph then the marginal distribution of any subset of the variables satisfies the same Markov property with respect to the corresponding induced subgraph. This fact follows immediately from the definition of the connected set Markov property because all the independence relationships implied by the subgraph are also encoded by the larger bi-directed graph.
It follows that if the distribution of $(Y_{P}, X_{J_{B}})$ satisfies the connected set Markov property with respect to $\G^{B}$ then, for every $j_{B}\in \J_{B}$, the random vector $(Y_{P}, X_{j_{B}})$ is a subvector of $(Y_{P}, X_{J_{B}})$ and, therefore, its distribution  satisfies the connected set Markov property with respect to the relevant induced subgraph $\G^{B}_{P\cup j_{B}}$. Hence, it is sufficient to show the inverse implication.

We denote by $L\subseteq P\cup J_{B}$ an arbitrary disconnected set in $\G^{B}$ and by $C_{1},\ldots, C_{r}$ its connected components. Furthermore, we set $Q=L\cap P$, $K=L\cap J_{B}$. Hence, we have to show that if the distribution of $(Y_{P}, X_{j_{B}})$ satisfies the connected set Markov property with respect to $\G^{B}_{P\cup j_{B}}$ for every $j_{B}\in \J_{B}$ then $C_{1}\ci\cdots\ci C_{r}$.

For $L\subseteq P\cup J_{B}$ we set $D^{*}=\{v \mid v\in B\mbox{ and } K\cap J_{v}\neq \emptyset\}$. Note that, if $L$ is primary we can write $L=Q\cup j_{D}$ by Lemma~\ref{THM:primary.sets} and, in this case, $D=D^{\ast}$. Furthermore, it is easy to check that $|K|\geq |D^{\ast}|$
with $|K|=|D^{\ast}|$ if and only if $L$ is primary. The proof is by induction on the value $(|K|-|D^{\ast}|)$.

We first assume $(|K|-|D^{\ast}|)=0$, which includes the case where $K=\emptyset$.
If $(|K|-|D^{\ast}|)=0$ then $L=Q\cup j_{D}$ is primary and there exists a $j_{B}\in \J_{B}$ such that $L=Q\cup j_{D}\subseteq P\cup j_{B}$. Hence, the independence $C_{1}\ci\cdots\ci C_{r}$ follows from the fact that, by assumption, $(Y_{P}, X_{j_{B}})$ satisfies the connected set Markov property with respect to $\G^{B}_{P\cup j_{B}}$.

We now show that if the result is true for $(|K|-|D^{\ast}|)\leq k$ with $k\geq 0$  then it is also true for $(|K|-|D^{\ast}|)=k+1$. If $(|K|-|D^{\ast}|)=k+1$ then  $(|K|-|D^{\ast}|)>0$ so that  $L$ is non-primary and, in turn, this implies that there exists a $v\in B$ such that $|J_{v}\cap K|\geq 2$. Hence, we set $J^{\prime}_{v}=J_{v}\cap K$ and remark that the subgraph induced by $J^{\prime}_{v}$ is complete because $\G^{B}$ is a $B$-expanded graph. This implies that $J^{\prime}_{v}$ is contained in exactly one connected component of $L$ and we assume, without loss of generality, that  $J^{\prime}_{v}\subseteq C_{1}$. For an arbitrary element $j^{\prime}_{v}\in J^{\prime}_{v}$ both $C_{1}\backslash \{j^{\prime}_{v}\},\ldots, C_{r}$ and $\{j^{\prime}_{v}\},\ldots, C_{r}$ are disconnected sets  in $\G^{B}$. Hence, it follows from  the induction assumption that both $C_{1}\backslash \{j^{\prime}_{v}\}\ci\cdots\ci C_{r}$ and $\{j^{\prime}_{v}\}\ci\cdots\ci C_{r}$ and by Lemma~\ref{THM:lemma.basic.indep} that  $C_{1}\ci\cdots\ci C_{r}$ as required.

\subsection*{Proof of Lemma~\ref{THM:zero.probability}}

Let $K$ be a primary subset of $J_{B}$. In this case, by Lemma~\ref{THM:primary.sets}, there exists a unique subset $D\subseteq B$ such that $K=j_{D}\in \J_{D}$ and, therefore, $X_{K}=X_{j_{D}}=(X^{j_{v}}_{v})_{j_{v}\in j_{D}}$.
In order to show that $\pr(X_{K}=1_{K}, X_{J_{B}\backslash K}=0_{J_{B}\backslash K})
=\pr(Y_{D}=j_{D}, Y_{B\backslash D}=0_{B\backslash D})$ we notice that (i) $J_{B}=J_{B\backslash D}\dcup J_{D}$ and therefore also $J_{B}\backslash j_{D}=J_{B\backslash D}\dcup (J_{D}\backslash j_{D})$ because $j_{D}\subseteq J_{D}$;
(ii) $X_{j_{D}}=1_{j_{D}}$ implies that $X_{J_{D}\backslash j_{D}}=0_{J_{D}\backslash j_{D}}$; (iii) it follows from (\ref{EQN:def.dichotomization}) that both
$X_{j_{D}}=1_{j_{D}}$ iff $Y_{D}=j_{D}$, and $X_{J_{B\backslash D}}=0_{J_{B\backslash D}}$ iff $Y_{B\backslash D}=0_{B\backslash D}$. Hence, by applying (i) to (iii) one obtains
\begin{eqnarray*}
\pr(X_{K}=1_{K}, X_{J_{B}\backslash K}=0_{J_{B}\backslash K})
&=& \pr(X_{j_{D}}=1_{j_{D}}, X_{J_{B}\backslash j_{D}}=0_{J_{B}\backslash j_{D}})\\
&=& \pr(X_{j_{D}}=1_{j_{D}}, X_{J_{B\backslash D}}=0_{J_{B\backslash D}}, X_{J_{D}\backslash j_{D}}=0_{J_{D}\backslash j_{D}})\\
&=& \pr(X_{j_{D}}=1_{j_{D}}, X_{J_{B\backslash D}}=0_{J_{B\backslash D}})\\
&=&\pr(Y_{D}=j_{D}, Y_{B\backslash D}=0_{B\backslash D}),
\end{eqnarray*}
as required.  The fact that for $j_{D}\neq \emptyset$ it holds that $\pr(X_{K}=1_{K})=\pr(Y_{D}=j_{D})$ follows from (iii) above. For every $K\subseteq J_{B}$ it holds by  (\ref{EQN:mobius.binary}) that $\tilde{\mu}_{K}=\pr(X_{K}=1_{K})$. Then it  follows from the result above and (\ref{EQN:mu.content}) that $\tilde{\mu}_{j_{D}}=\pr(X_{j_{D}}=1_{j_{D}})=\pr(Y_{D}=j_{D})=\mu^{j_{D}}$. By definition, $\tilde{\gamma}_{K}=\sum_{E \subseteq K} (-1)^{|K \backslash E|} \log\tilde{\mu}_{E}$ and if  $K=j_{D}$ is a primary subset of $J_{B}$ then
it is not difficult to check that (i) the set $\{E| E\subseteq j_{D}\}$ can be alternatively written in the form $\{j_{H}=(j_{D})_{H}| H\subseteq D\}$, (ii) if $H\subseteq D$ then $(j_{D})_{H}=j_{H}\in \J_{H}$ so that $\tilde{\mu}_{j_{H}}=\mu^{j_{H}}$, and (iii) for every $H\subseteq D$ it holds that $|j_{D}|=|D|$, $|j_{H}|=|H|$ and $j_{H}\subseteq j_{D}$ so that $|j_{D} \backslash j_{H}|=|D \backslash H|$. By using points (i) to (iii) above one obtains
\begin{eqnarray*}
\tilde{\gamma}_{j_{D}}
&=& \sum_{E \subseteq j_{D}} (-1)^{|j_{D} \backslash E|} \log\tilde{\mu}_{E}\\
&=&\sum_{H \subseteq D} (-1)^{|j_{D} \backslash j_{H}|} \log\tilde{\mu}_{j_{H}}\\
&=& \sum_{H \subseteq D} (-1)^{|D \backslash H|} \log\mu^{j_{H}}\\
&=&\gamma^{j_{D}},
\end{eqnarray*}
and this completes the proof of the first statement.

We prove the second statement by  noticing that, if $K\subseteq J_{B}$ is non-primary, then, it follows from the definition of primary subset that there exists a $v\in B$ such that $|K\cap J_{v}|>1$. Consequently, there exists a pair $j_{v},j^{\prime}_{v}\in J_{v}$ with $j_{v}\neq j^{\prime}_{v}$ such that $j_{v},j^{\prime}_{v}\in K$. In this case, $\pr(X_{K}=1_{K})\leq \pr(X^{j_{v}}_{v}=1, X^{j^{\prime}_{v}}_{v}=1)=\pr(Y_{v}=j_{v}, Y_{v}=j^{\prime}_{v})=0$. This also implies that $\pr(X_{K}=1_{K}, X_{J_{D}\backslash K}=0_{J_{B}\backslash K})=0$ because $\pr(X_{K}=1_{K}, X_{J_{D}\backslash K}=0_{J_{B}\backslash K})\leq\pr(X_{K}=1_{K})=0$. It follows also that $\tilde{\mu}_{K}=\pr(X_{K}=1_{K})=0$ and that  $\tilde{\gamma}_{K}$ is well-defined because its computation involves the logarithm of $\tilde{\mu}_{K}=0$.

\subsection*{Proof of Theorem~\ref{THM:cor.to.binary.exp.prop}}
We first consider the case where $L$ is a primary subset of $P\cup J_{B}$, so that $L=Q\cup j_{D}$ as in Lemma~\ref{THM:primary.sets}.
If we take the binary expansion of $Y_{V}$ with respect to $J_{V}=J_{P}\dcup J_{B}$, that is equal to $X_{J_{V}}=(X_{J_{P}}, X_{J_{B}})$, then it follows from Lemma~\ref{THM:zero.probability} that %
\begin{eqnarray}\label{EQN:002}
\mu^{j_{Q\cup D}}=\tilde{\mu}_{j_{Q\cup D}}\qquad\mbox{and}\qquad \gamma^{j_{Q\cup D}}=\tilde{\gamma}_{j_{Q\cup D}}
\end{eqnarray}
for every $j_{Q}\in \J_{Q}$. Recall that, since both the \Mob\ and the  \LML\ parameters satisfy the upward compatibility property then the computation of both  $\tilde{\mu}_{j_{Q\cup D}}$ and $\tilde{\gamma}_{j_{Q\cup D}}$ can be carried out with respect to the distribution of $(X_{j_{Q}}, X_{j_{D}})$ that is a subvector of $(X_{J_{Q}}, X_{j_{D}})$. It also follows from upward compatibility that, for every $j_{Q}\in \J_{Q}$, both $\tilde{\mu}_{j_{D}}^{j_{Q}}$ and $\tilde{\gamma}_{j_{D}}^{j_{Q}}$ can be computed with respect to the distribution of $(Y_{Q}, X_{j_{D}})$. One should notice that $X_{j_{D}}$ contains exactly one entry for every variable in $Y_{D}$ and, as a consequence, the probability table of $(Y_{Q}, X_{j_{D}})$ is strictly positive because it can be obtained by collapsing some levels of the probability table of $Y_{Q\cup D}$, that is assumed to be strictly positive. For this reason it makes sense to consider $(Y_{Q}, X_{j_{D}})$ in place of $Y_{V}$ in Lemma~\ref{THM:zero.probability} to show that
\begin{eqnarray}\label{EQN:011}
\tilde{\mu}_{j_{D}}^{j_{Q}}=\tilde{\mu}_{j_{Q\cup D}}
\qquad\mbox{and}\qquad
\tilde{\gamma}_{j_{D}}^{j_{Q}}=\tilde{\gamma}_{j_{Q\cup D}}
\end{eqnarray}
where the quantities $\tilde{\mu}_{j_{Q\cup D}}$ and $\tilde{\gamma}_{j_{Q\cup D}}$  in (\ref{EQN:011})  coincide with same quantities in (\ref{EQN:002}) because they are computed with respect to the binary expansion $(X_{J_{Q}}, X_{j_{D}})$ of $(Y_{Q}, X_{j_{D}})$. Hence, (\ref{EQN:002}) and (\ref{EQN:011}) lead to  $\mu^{j_{Q\cup D}}=\tilde{\mu}_{j_{D}}^{j_{Q}}$ and $\gamma^{j_{Q\cup D}}=\tilde{\gamma}_{j_{D}}^{j_{Q}}$ as required.

Consider now the case where $L$ is non-primary. If we set $K=L\cap J_{B}$ and $Q=L\cap P$, then, for every $j_{Q}\in \J_{Q}$ it follows form  definition of \Mob\ parameter and Lemma~\ref{THM:zero.probability} that $\tilde{\mu}_{j_{D}}^{j_{Q}}=\pr(Y_{Q}=j_{Q}, X_{K}=1_{K})\leq \pr(X_{K}=1_{K})=0$. This implies that $\tilde{\gamma}_{j_{D}}^{j_{Q}}$ is not well-defined because its computation involves the logarithm of $\tilde{\mu}_{j_{D}}^{j_{Q}}=0$.

\subsection*{Proof of Theorem~\ref{THM:markov-property-for-B-expanded}}
We start this proof by remarking that, by Lemma~\ref{THM:primary.sets}, a subsets $L\subset P\cup J_{B}$ is such that $L\subset P\cup j_{B}$ for some $j_{B}\in \J_{B}$ if and only if it is a primary subset of $P\cup J_{B}$. This also implies that if $L\subset P\cup j_{B}$ we can write $L=Q\cup j_{D}$, as in Lemma~\ref{THM:primary.sets}.

By Lemma~\ref{EQN:lemma.expanded.Markov}, the distribution of $Y^{B}=(Y_{P}, X_{J_{B}})$ satisfies the connected set Markov property with respect to $\G^{B}$ if and only if the distribution of $(Y_{P}, X_{j_{B}})$ satisfies the connected set Markov property with respect to $\G^{B}_{P\cup j_{B}}$ for every $j_{B}\in\J_{B}$.

For every $j_{B}\in \J_{B}$, the random vector $X_{j_{B}}$ contains exactly one entry for every variable in $Y_{B}$ and, as a consequence, the probability table of $(Y_{P}, X_{j_{B}})$ is strictly positive because it can be obtained by collapsing some levels of the probability table of $Y_{V}$, that is strictly positive by assumption. Hence, we can apply Corollary~\ref{THM:gamma.eq.zero.and.disc.MP}, and it follows that the distribution of $(Y_{P}, X_{j_{B}})$ satisfies the connected set Markov property with respect to $\G^{B}_{P\cup j_{B}}$ if and only if for every set $L\subseteq P\cup j_{B}$ that is disconnected in $\G^{B}_{P\cup j_{B}}$ it holds that $\tilde{\gamma}^{j_{Q}}_{j_{D}}=0$ for every $j_{Q}\in \J_{Q}$. Note that it is possible that either $Q=\emptyset$ or $D=\emptyset$. We can thus state that:
\begin{quote}
The distribution of $Y^{B}=(Y_{P}, X_{J_{B}})$ satisfies the connected set Markov property with respect to $\G^{B}$ if and only if for every $j_{B}\in \J_{B}$ and  every set $L\subseteq P\cup j_{B}$ that is disconnected in $\G^{B}_{P\cup j_{B}}$ it holds that $\tilde{\gamma}^{j_{Q}}_{j_{D}}=0$ for every $j_{Q}\in \J_{Q}$.
\end{quote}
Notice that: (a) a subsets $L\subseteq P\cup J_{B}$ is such that $L\subseteq P\cup j_{B}$ for some $j_{B}\in \J_{B}$ if and only if it is a primary subset of $P\cup J_{B}$, (b) $L\subset P\cup J_{B}$ is disconnected in $\G^{B}_{P\cup j_{B}}$ if and only if it is disconnected in $\G^{B}$ and (c)  $\tilde{\gamma}^{j_{Q}}_{j_{D}}=\gamma^{j_{Q\cup D}}$ by Corollary~\ref{THM:cor.to.binary.exp.prop}. Hence, by using (a), (b) and (c) the statement above can be rephrased as:
\begin{quote}
The distribution of $Y^{B}=(Y_{P}, X_{J_{B}})$ satisfies the connected set Markov property with respect to $\G^{B}$ if and only if for every set $L\subseteq P\cup J_{B}$ that is a primary subset of $P\cup J_{B}$ and disconnected in $\G^{B}$ it holds that $\gamma^{j_{Q\cup D}}=0$ for every $j_{Q}\in \J_{Q}$,
\end{quote}
and this completes the proof.
%

%
\bibliographystyle{chicago}
\bibliography{gamma-ref-imp-dic}
\appendix
\end{document}